\documentclass[11pt]{article}
\usepackage[reqno,tbtags]{amsmath}
\usepackage{fullpage}
\usepackage{amsthm}
\usepackage{amsfonts}
\usepackage{enumerate}

 \DeclareMathOperator{\E}{E}

\newcommand{\argmin}{\mathrm{argmin}}

\newcommand{\bGamma}{\boldsymbol{\Gamma}}
\theoremstyle{plain}
\newtheorem{thm}{Theorem}[section]
\newtheorem{cor}[thm]{Corollary}
\newtheorem{lem}[thm]{Lemma}
\newtheorem{prop}[thm]{Property}
\theoremstyle{definition}

\theoremstyle{remark}

\usepackage{epsfig}
 \usepackage{ifpdf}
 \ifpdf\fi
\numberwithin{equation}{section}

 \DeclareMathOperator{\glex}{\prec}
 \DeclareMathOperator{\geqlex}{\preceq}
 \DeclareMathOperator{\llex}{\succ}

\let\olditemize=\itemize
\def\itemize{
\olditemize
\setlength{\itemsep}{-1ex}
}
\let\oldenumerate=\enumerate
\def\enumerate{
\oldenumerate
\setlength{\itemsep}{-1ex}
}

\newcommand{\ignore}[1]{}

\begin{document}
\title{A Labeling Approach to Incremental Cycle Detection}

\author{Edith Cohen\thanks{Microsoft Research -- SVC, USA. {\tt
      edith@cohenwang.com}} $^\dagger$  \and Amos Fiat\thanks{Tel Aviv
    University, Tel Aviv, Israel. {\tt  \{fiat,haimk\}@cs.tau.ac.il}}
\and Haim Kaplan$^\dagger$ \and Liam
Roditty\thanks{Bar Ilan University, Ramat Gan, Israel. {\tt liam.roditty@biu.ac.il}}}

\maketitle

\begin{abstract}
In the \emph{incremental cycle detection} problem arcs are added to
a directed acyclic graph and the algorithm has to report if the new
arc closes a cycle. One seeks to minimize the total time to process
the entire sequence of arc insertions, or until a cycle appears.

In a recent breakthrough, Bender, Fineman, Gilbert and
Tarjan~\cite{BeFiGiTa11} presented two different algorithms, with
time complexity  $O(n^2 \log n)$ and $O(m \cdot \min \{ m^{1/2},
n^{2/3} \})$, respectively.

In this paper we introduce a new technique for incremental cycle
detection that allows us to obtain both bounds (up to a logarithmic
factor). Furthermore, our approach seems more amiable for
distributed implementation.
\end{abstract}\textbf{}

\thispagestyle{empty}

\newpage
\setcounter{page}{1}

\section{Introduction}


Let $G=(V,E)$ be a directed acyclic graph (DAG). In the
\emph{incremental cycle detection} problem edges are being added to
$G$ and the algorithm has to report on a cycle once a cycle is
formed.

This problem has a very extensive history
\cite{AlpernHRSZ90,Marchetti-SpaccamelaNR96,PearceK06,KatrielB06,LiuC07,AjwaniFM08,AjwaniF10,DBLP:journals/corr/abs-0711-0251,HKMST:Talg12,BFG:SODA09}.
For a thorough discussion of this work see~\cite{HKMST:Talg12}.  In
a recent breakthrough, Bender, Fineman, Gilbert and
Tarjan~\cite{BeFiGiTa11} presented an algorithm with $O(n^2 \log n)$
total running time. They also presented a different algorithm with a
running time of $O(m \cdot \min \{ m^{1/2}, n^{2/3} \})$.

In this paper we present a new and completely different technique
that allows us to obtain all the results of Bender et al. (up to
poly-logarithmic factors and randomization). Although we are not
getting any improved running times our technique is interesting from
several perspectives. We believe that our approach unifies all
previous algorithms into one algorithmic framework. Furthermore, our
algorithm seems (to us) much simpler than previous proposals.
Finally, because of highly local nature, it seems that it is trivial
to implement our incremental cycle detection algorithm in a
distributed environment, within certain caveats.

Roughly speaking, our framework works in the following way. As long
as a cycle is not formed we maintain a certain label $\ell(v)$ for
each vertex $v$ so that the labels constitute  a {\em weak
topological order}: That is, for every arc $(u,v)$, $\ell(u) \prec
\ell(v)$. These labels are useful to rule out the existence of paths
from a vertex $y$ to a vertex $x$ if $\ell(y) \succ \ell(x)$.

The development of our new labeling technique is inspired by the
work of Cohen~\cite{Cohen:1997} on estimating the size of the
transitive closure of a directed graph. More specifically,
Cohen~\cite{Cohen:1997} showed that if the vertices of an $n$ vertex
digraph  get random ranks from the range $1,\ldots,n$ then the
minimum rank vertex that can reach to every vertex $u$ can be
computed in $O(m)$ time for every $u\in V$. We can view the rank of
the minimal rank vertex that reaches $u$ as the label of $u$.
This label is a good estimate of the number of vertices that reach
 $u$. A label of small value indicates that the reachability set
is probably large.

In this paper we give a recursive version of the labels described
above. Let $0 < q \leq 1$. Given an $n$ vertex DAG, assign a random
permutation of the ranks $1,\ldots,qn$ to a randomly selected set of
$qn$ vertices. Other vertices are {\sl unranked}. The label of a
vertex $u$ is defined to be a sequence of vertices, the first of
which, $\ell_1(u)$, is the vertex of minimal rank amongst all ranked
vertices that can reach $u$.  The second vertex in the label of $u$
is the vertex of minimal rank amongst all ranked vertices $v\neq
\ell_1(u)$, such that $v$ is reachable from $\ell_1(u)$ and $u$ is
reachable from $v$. Subsequent vertices in the label of $v$ are
defined analogously. Notice that the first coordinate of each label
in our extended definition is the label from~\cite{Cohen:1997}.

Such random recursive labels have several properties of possible
interest:
\begin{enumerate}
  \item The expected length of such labels is logarithmic.
\item For every vertex $u$, consider the sequence of ranks associated with the vertices in the label of $u$, with $\infty$ appended at the end. Such sequences are lexicographically descending along any path through a DAG. Thus, in certain cases comparison of two labels can rule out the existence of a path.
\item For any vertex $u$, the set of vertices $v$ such that $u$ and $v$ have the same label and $u$ is reachable from $v$ is ``small". With high probability this set is $O(\log n/q)$.
\end{enumerate}

As in several previous papers, we do both forward and backwards
searches to determine if a cycle has been formed. One difference
between previous approaches and ours is that local criteria allow us
to prune both forward and backward searches. The labels contain
sufficient information so as to make this pruning efficient.
Moreover, the labels can be maintained over the sequence of insertions within the same
time bounds.

%
%


Using the labels, setting appropriate parameters, and some simple
data structures, we get a family of possible algorithms, which
unifies the results of several previous papers:

\begin{itemize}
\item For graphs with $m=O(n)$, choosing $q=1/\sqrt{n}$ gives us an algorithm with total time $O(n^{3/2}\log n)$.
For denser graphs, we draw a random rank for each arc with
probability $1/\sqrt{m}$ and set the rank of the vertex to be the
rank of its minimum incoming arc. This gives the analogous bound of
$O(m^{3/2}\log n)$.

\item Choosing $q=\sqrt[3]{\log(n)/n}$ balances the forward and backward search times to be $O(m \cdot n^{2/3}\log^{4/3} n)$.
\item Choosing $q=1$ requires no backward search, as vertices have
  unique labels, and the total time for forward searches and label updates is $O(n^2\log^2 n)$.

\end{itemize}

All of these variants can be implemented using message passing
algorithms, if one allows bidirectional communications and one
assumes perfect synchrony. With respect to distributed
implementation, it is often important to minimize the number of
messages.
 We remark that for each of these
variants, one can optimize $q$ so as to minimize the number of
messages.


\smallskip
\noindent {\bf Related work:}

A directed graph is acyclic if and only if it has a topological
order; a more recent generalization is that the strong components of
a directed graph can be ordered topologically
\cite{harary1965structural}. We can find a cycle in a directed graph
or a topological order in linear time either by repeatedly deleting
 vertices with no
predecessors \cite{KnuthS74} or by performing a depth-first search
\cite{Tarjan72dfs}. Depth-first search can also be used to find the
strong components and a topological order of these components in
$O(m)$ time \cite{Tarjan72dfs}.

  The
  digraph cycle detection problem has an extensive history
  \cite{AlpernHRSZ90,Marchetti-SpaccamelaNR96,PearceK06,KatrielB06,LiuC07,AjwaniFM08,AjwaniF10,DBLP:journals/corr/abs-0711-0251,HKMST:Talg12,BFG:SODA09}.
    The current state-of-the-art
time bounds by centralized algorithms are  $O(m^{3/2})$ by Haeupler
 et al.\ \cite{HKMST:Talg12}
 and $O(n^2\log n)$ by Bender  et al.\ \cite{BFG:SODA09}.
  The two-way search  algorithm of  Haeupler
 et al.\ maintains a complete topological order.
    When
an insertion occurs which is inconsistent with the order,  nodes are
shuffled to correct this.

 Bender et
al.\ \cite{BFG:SODA09} suggested a simpler algorithm that runs in
 $O(m \cdot \min \{ m^{1/2}, n^{2/3} \})$ time. This algorithm maintains only a
 weak topological order. It partitions
 the vertices  into levels and when an arc is inserted it performs a
 backward search within a level and a forward search across levels.
 This algorithm stops the backward search when it reaches a
prespecified number of arcs.

There has been little work on distributed cycle detection, even when
the graph is static. Fleischer et al.\ \cite{FleischerHP00}
suggested a divide and conquer based randomized algorithm for
finding strongly connected components that is easier to parallelize
and sequentially runs in expected $O(m\log n)$ time. The distributed
cycle detection problem also arises in the context of model checking
on large flow graphs. Barnat et al.\ \cite{BarnatBC05} gave a
distributed algorithm based on breadth first search which is
quadratic in the worst case. A distributed implementation of our
algorithm, which is subquadratic, is interesting even for a static
graph.

\smallskip
\noindent {\bf Organization of this paper:}

In Section \ref{prelim:sec} we give basic definitions and properties
of our labeling. In Section \ref{labels:sec} we consider the case
when ranked vertices and ranks are determined probabilistically and
give some properties that hold in expectation and with high
probability.
 In Section
\ref{sec:2directions} we present a dynamic algorithm for maintenance
of labels and analyze its complexity (time and message complexity).
In Section \ref{nsquare:sec} we give a variant of the dynamic
algorithm that gives us the $O(n^2\log^2 n)$ time result.

\section{Preliminaries} \label{prelim:sec}

Let $G=(V,A)$ be a directed acyclic graph with vertices $V$ and arcs $E$, $|V|=n$ and  $|A|=m$.
We define  $P(v)$ to be the set of all predecessors of $v$ ({\sl
i.e.}, for all $u\in P(v)$ there is a path from $u$ to $v$),  and
$S(v)$ to be the set of all successors of $v$ ({\sl i.e.}, for all
$u\in S(v)$ there is path from $v$ to $u$). We include $v$ in its
predecessors and successors sets, that is, $v\in P(v)\cap
S(v)$. Also, define $$D(u,v) = S(u) \cap P(v).$$

Let $L$ be a subset of $V$, $|L|=\lambda\leq n$.  Let $r:V\mapsto
Z^+\cup \{\infty\}$  be such that $r_{|L}:L\mapsto Z^+$ is one to
one and $r(v)=\infty$ for all $v\in V-L$. We say that the
vertices in $L$ are {\em ranked}.

We define $\overline{A}_1(v) \equiv P(v)$ and $A_1(v) = P(v)\cap L$, the set of ranked predecessors of
$v$. If $A_1(v) \not= \emptyset$ then we define $\ell_1(v) =
\argmin_{u\in A_1(v)} r(u)$, 
$\overline{A}_2(v)=  D(\ell_{1}(v),v)
\setminus \{\ell_1(v)\}$, and  $A_{2}(v) = \overline{A}_2(v) \cap L$. If $A_2(v)
\not= \emptyset$ then we define $\ell_2(v) = \argmin_{u\in A_2(v)}
r(u)$, $\overline{A}_3(v)=  D(\ell_{2}(v),v) \setminus \{\ell_2(v)\}
$, and  $A_{3}(v) = \overline{A}_3(v) \cap L$.  We continue in the same way and for all $i
\ge 1$ such that $A_i(v) \not= \emptyset$ we define $\ell_i(v)$, $\overline{A}_{i+1}(v)$,  and
$A_{i+1}(v)$. We define $k(v)$ to be $0$ if $A_1(v) = \emptyset$ and
 we define $k(v)$ to be the largest $i$ for which $A_i(v)
\not= \emptyset$.


 We define the
\emph{label} of a vertex $v\in V$ to be the sequence $\ell(v) =
\ell_1(v), \ell_2(v), \ldots, \ell_{k(v)}(v).$ Figure~\ref{graph:fig} in the Appendix presents an example of this labeling.

The following
properties stem directly from the definitions above:
\begin{prop}
For all $v\in V$,
\begin{itemize}
\item For all $1 < i \leq k(v)$, $r(\ell_{i-1}(v)) < r(\ell_i(v))$.
\item For every $v\in V$ such that $r(v)<\infty$,
 $v\in A_1(v)\cap A_2(v) \cap \cdots \cap A_{k(v)}(v)$ and $\ell_{k(v)}(v)=v$.
\end{itemize}
\end{prop}

We now define a partial order $\prec$ of the labels.  The order
$\prec$  corresponds to a
{\em decreasing} lexicographic order on $$r(\ell(v)) \equiv (r(\ell_1(v)),r(\ell_2(v)),\ldots,
r(\ell_{k(v)}(v)), +\infty)\ .$$
That is, for $u,v\in V$, $\ell(v) \prec \ell(u) \iff r(\ell(v)) \raisebox{-3pt}{$\overset{>}{\mbox{\rm\tiny lex}}$}  r(\ell(u))$,
which happens if either of the following holds:
\begin{description}
\item[Case 1:] For some $1\leq i\leq k(u)$ we have that $r(\ell_i(u))>r(\ell_i(v))$, and for all $1 \leq j < i$: $\ell_j(u) = \ell_j(v)$, or,
\item[Case 2:] $k(u)<k(v)$ and for all $1 \leq j \leq k(u)$:
$\ell_j(u)=\ell_j(v)$.
\end{description}

For $u,v \in V$, define $\mathrm{LCP}(u,v)$ to  be the longest
common prefix of $\ell(u)$ and $\ell(v)$. By definition, both $u$
and $v$ are reachable from every vertex in $\mathrm{LCP}(u,v)$. In
particular, if $\ell(u)$ is a prefix of $\ell(v)$ and $u\in L$
(which implies that $u\in \ell(u)$) then $v$ is reachable from $u$.
Next, we show that in certain cases the labels can be used to rule out the existence of a path. This will be used later on by our cycle detection algorithm.

\begin{thm}\label{T-no-path}
If $\ell(u) \glex \ell(v)$ then there is no path from $v$ to $u$.
\end{thm}
\begin{proof}
Assume that case 1 holds and let $i$ be minimal such that
$r(\ell_i(u))>r(\ell_i(v))$. Note that $\ell_i(u)$ is the vertex of
minimal rank in $A_i(u)$ and $\ell_i(v)$ is the vertex of minimal
rank in $A_i(v)$. If there is a path from $v$ to $u$ then it must be
that $A_i(v)\subset A_i(u)$, which implies that $r(\ell_i(u)) \leq
r(\ell_i(v))$, a contradiction.

Now, assume that case 2 holds and consider the vertex $w=\ell_{k(u)}(u)$ ($=\ell_{k(u)}(v)$), and the vertex $w'=\ell_{k(u)+1}(v)$. There is a path from $w$ to $w'$,  and there is a path
from $w'$ to $v$. If there is also a path from $v$ to $u$ we have that $w'\in A_{k(u)+1}(u)$, but $A_{k(u)+1}(u) = \emptyset$, a contradiction.
\end{proof}

\begin{figure}[t]
\centering
\ifpdf
\includegraphics[width=0.4\textwidth]{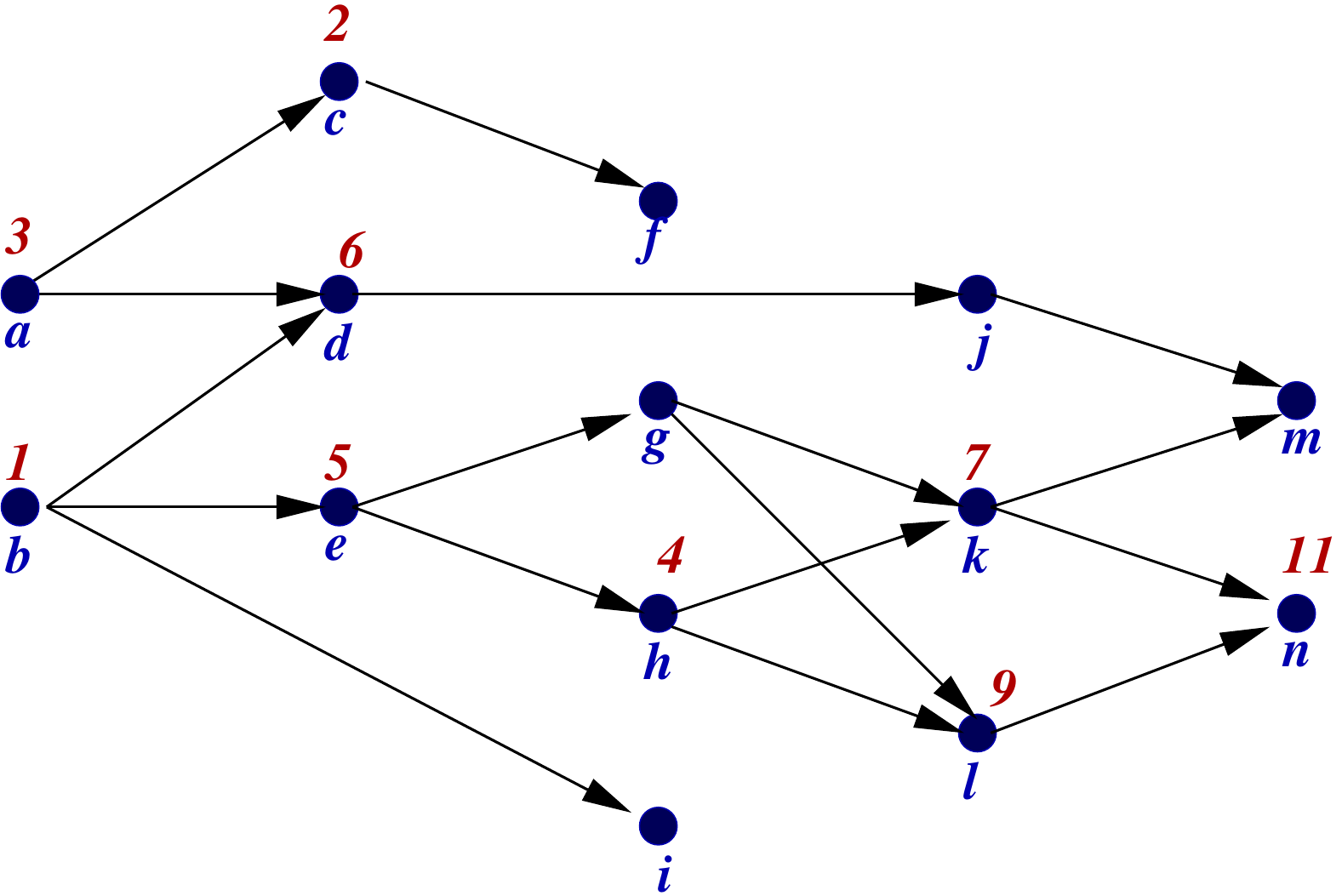}
\else
\epsfig{figure=graph.eps,width=0.4\textwidth}
\fi

\caption{Example graph. The nodes with numbers are the ranked nodes $L=\{a,b,c,d,e,h,k,l,n\}$.  For node $n$ we have $A_1(n)= \{b,e,h,k,l,n\}$, $\ell_1(n)=b$, $A_2(n)=\{e,h,k,l,n\}$, $\ell_2(n)=h$, $A_3(n)=\{k,l,n\}$, $\ell_3(n)=k$,  $A_4(n)=\{n\}$, $\ell_4(n)=n$, $A_5(n)=\emptyset$.  For node $f$ we have $A_1(f)=\{a,c\}$, $\ell_1(f)=c$, $A_2(f)=\emptyset$.
Therefore $k(f)=1$ and $k(n)= 4$. \label{graph:fig}}
\end{figure}

\section{Properties of labels}  \label{labels:sec}

In this section we give properties of our new labelings when ranks are assigned at random. In particular we show that
\begin{itemize}
  \item Labels are ``short'', see Lemma \ref{L-label-size}.
  \item The set of predecessors of a vertex $v$ that have the same label as $v$ is ``small'', see Lemma \ref{lem:backsize}.
\end{itemize}
These properties are used in Section 4 for the analysis of our algorithms.

Let $L$ be a subset of $V$, $|L| = \lambda \le n$, and let $R$ be
the set of all one to one mappings $r:L\mapsto
\{1,\ldots,\lambda\}$. In this section we first analyze the size of
the labels when $r$ is chosen uniformly at random from $R$.

Let $v\in V$, every mapping $r\in R$ determines  $A^r_i(v)$,
$\ell^r_i(v)$,  for $1 \leq i$ and $k^r(v)$. We omit the mapping $r$
when clear from the context (as done above). For any sequence of
subsets of ranked vertices  $\bGamma = \langle \Gamma_1 \supset
\Gamma_2 \supset \dots \supset \Gamma_{j}\rangle$  define the subset
of mappings
$$R(v,\bGamma,j) = \{r\in R \mid A^r_i(v) = \Gamma_i \mbox{\rm\ for
all $1 \leq i \leq j$}\}.$$

We now show:
\begin{lem}\label{L-set-size} For any  $v\in V$, any $x\in [0,1]$,
 and a sequence of
subsets of ranked vertices  $\bGamma = \langle \Gamma_1 \supset
\Gamma_2 \supset \dots \supset \Gamma_{j}\rangle$   such that
$R(v,\bGamma,j) \not=\emptyset$ we have that
$$\Pr_{r\sim
  R(v,\bGamma,j)}\left[\frac{|A^r_{j+1}(v)|}{|\Gamma_{j}|} \leq
  x\right]   \geq x\ .$$
\end{lem}
\begin{proof}
 The set of mappings $R(v,\bGamma,j)$ are partitioned into
 $|\Gamma_{j}|!$ equal size equivalence classes. Each class
contains  all the rankings that induce the same
 relative order of the ranks
$r(s)$, $s\in \Gamma_{j}$. That is,  rankings $r_1$ and $r_2$ are in
the same class  iff for any pair of vertices $x,y \in \Gamma_j$,
$r_1(x) < r_1(y) \Leftrightarrow r_2(x) < r_2(y)$.

Consider some (arbitrary, fixed) topological order on $\Gamma_{j}$,
$U=\langle u_1, u_2, \ldots, u_{|\Gamma_{j}|}\rangle$. For $r\in
R(v,\bGamma,j)$ let $m(r,U) = \argmin_i r(u_i)$, {\sl i.e.}, the
position in the topological order of the vertex in $\Gamma_j$ with minimum rank value.
It follows that for all $1 \leq i \leq |\Gamma_{j}|$,
   $$\Pr_{r\sim R(v,\bGamma,j)} [ m(r,U) = i ] = \frac{1}{|\Gamma_{j}|}.$$

It follows that
  \begin{eqnarray*}
\Pr_{r\sim R(v,\bGamma,j)} \left[|A^r_{j+1}(v)| \leq x |\Gamma_{j}|
\right] &\geq& \Pr_{r\sim R(v,\bGamma,j)} \left[ m(r,U) \geq (1-x)
  |\Gamma_{j}| \right] \geq   x \ .
\\ \end{eqnarray*}
\end{proof}

 We are now
ready to bound the label size.

\begin{lem}\label{L-label-size}
For any $c\geq e$,  the probability that there exists $u\in V$ whose
label has more than $c \log \lambda$ vertices is at most 
$\min\{\lambda, \lambda^{-1-c (\ln c -1)}\}$.
 \end{lem}
\begin{proof}
Consider tne distribution on the ratio $|A_{j+1}|/| A_j|$.  From
Lemma~\ref{L-set-size}, 
this distribution is
dominated by the uniform disribution $U[0,1]$, in the sense that
$\forall x\in [0,1],  \Pr[|A_{j+1}|/|A_j| \leq x]  \geq x$.  The
distribution on the logarithm of the ratio is dominated by  $-\ln u$,
where $u\sim U[0,1]$.  This is an exponential distribution with
parameter $1$.  We now consider the product of these ratios over $j=1,\ldots,k$,
which is the ratio $|A_{k+1}|/|A_1|$.
The negated logarithm of the product, $-\ln  (|A_{k+1}|/|A_1|)$  is
the sum of the negated  logarithms of the
ratios $|A_{j+1}|/|A_j$ , which is dominated by the random variable $S_k$ that is the
sum of $k$ i.i.d exponential random variables.  This is a gamma distribution
which has cummulative distribution function (CDF) 
$\frac{x^{k-1}\exp(-x)}{k!}$.  From this domination relation, it
follows that the
probability that the label size is more than $k$ is at most 
$\Pr[S_k \leq \ln \lambda]$.  Substituting $x=\ln \lambda$ in the
CDF  we obtain the bound
$\frac{1}{\lambda k!} (\ln\lambda)^{k-1}$.  We substitute  $k= c
\ln\lambda$ and use the Stirling bound
$k! \geq  \sqrt{2\pi k} \left(\frac{k}{e}\right)^k$
obtaining, for $\lambda,c > e$, a bound of $\lambda^{-1-c (\ln c -1)}$.
To complete the proof for $\lambda \leq e$, 
note that the label size can be at most $\lambda$.
\end{proof}

The previous lemma applies for any set $L$ of labeled vertices. We now
consider the setting in which the set $L$ itself is also chosen at
random. The labels of the vertices are computed as before.
 Given a digraph $G=(V,E)$, a $q$-labeling of $G$ is the following:
 \begin{itemize}
 \item Each  vertex chooses to be ranked independently with probability $q$, {\sl i.e.}, $\E[\lambda]=qn$.
 \item The ranks of the ranked vertices is a random permutation
  of $1,\ldots,\lambda$ (by choosing ranks at random from a set of size $\lambda^c$, we can assume that the ranks are distinct and the relative order between them is a random permutation).
 \end{itemize}

 Given any graph $G$, any labeling $\ell$, and for any $v\in V$, let $I(G,\ell,v) = \{u\in P(v)\mid \ell(u)=\ell(v), \mbox{\rm\ $u$ not ranked}\}$, note that $|I(G,\ell,v)| \leq |P(v)|$ for all $G$, $\ell$, and $v$.

 Given $G$ and labeling $\ell$, define $I(G,\ell)$ to be the maximum over $v$ of $I(G,\ell,v)$.

 \begin{lem} \label{lem:backsize}
   For any graph $G$ with $n$ vertices, $0<q\leq 1$, and $c\geq
   3$, $$\mbox{\rm Prob}_{\ell \sim Q}[ I(G,\ell) \leq c\log(n) /q ] \geq 1-1/n^{c-2},$$ where the distribution $Q$ is the space of $q$-labelings.
 \end{lem}
 \begin{proof}
   Consider a vertex $v$ with label $\ell(v)= v_1,v_2, \ldots, v_t$. The set $D(v_t,v)$ is exactly the set of unranked vertices such that $\ell(u)=\ell(v)$ and there is a path from $u$ to $v$.
   Let $s=|D(v_t,v)|$. Every vertex $u\in D(v_t,v)\setminus \{v_t\}$ must be unranked. To see this, consider what happens if $u$ had rank, and $r(u)<r(v_t)$ then $v_t$ will not appear in $\ell(v)$, if $u$ had rank and $r(u)> r(v_t)$ then $v_t$ would not be last in $\ell(v)$.

  The probability that no vertex $u\in D(v_t,v)$ has rank is $(1-q)^s$. For $s>c\log(n)/q$ we have $(1-q)^s\leq n^{-c}$.
  From the union bound, the probability that no vertex $v$ has $|D(v_t,v)|>c \log(n)/q$ is at most $n^{-c} \cdot n^2= 1/n^{c-2}$ as there are at most $n \choose 2$ possible pairs $v$, $v_t$.

\end{proof}

\begin{cor}\label{cor:backsize}
  Over any sequence of insertions of edges, amongst $n$ vertices, resulting in graphs $G_1, G_2, \ldots, G_m$, $$\mbox{\rm Prob}_{\ell \sim Q}\left(\max_{i=1\ldots m} I(G_i,\ell)>c \log(n)/q\right) \leq 1/n^{c-4}.$$
\end{cor}
\begin{proof}
  This follows again from the union bound and as $m \leq n^2$.
\end{proof}

\section{Preserving labels dynamically} \label{sec:2directions}

Consider a set of vertices $V$ and let $e_1=(v_1,u_1), \ldots,
e_t=(v_t,u_t)$ be a sequence of arcs inserted over time. We seek to
maintain the labels $\ell(v)$ defined above, over this sequence of
insertions. (Extending the algorithm to allow addition of new
singleton vertices is straightforward.)

When adding an arc $(u,v)$, we update all labels if no cycle was
formed. If a cycle is created we halt.

\emph{Insert($u,v$):}
\begin{itemize}
\item If $\ell(u) \glex \ell(v)$ then we do nothing and return.
\item Cycle-Detect($u,v$).
\item If no cycle is detected and $\ell(u)\llex \ell(v)$ call Update($u,v$), where Update($x,y$) is a recursive procedure defined below.
\end{itemize}

\emph{Cycle-Detect($u,v$):}
\begin{itemize}
\item {\bf Backward search:} Starting from $u$,
send a message $msg=\langle v,\ell(u) \rangle$ to all in-neighbors of $u$. When a vertex $w$ gets such a
message $msg$ over an edge $e$ then it performs one of the
followings.\\
1) If $w\not=v$,  $\ell(w)=\ell(u)$ and $w$ gets $msg$ for the first time, then $w$ sends  $msg$ to all its in-neighbors. \\
2) If $\ell(w) \glex \ell(u)$ or $w$ has no in-neighbors then $w$
sends a ``no-cycle'' message back over the edge $e$.\\
3) If $w=v$ then $v$ sends a ``cycle'' message back on the edge $e$.

When $w$ gets a ``cycle'' message  for the first time then it
forwards it to one of its out-neighbors from which it got $msg$ and
stop sending messages. When $w$ gets a ``no-cycle'' message from all
 its in-neighbors then it sends a ``no-cycle'' message to all its
out-neighbors from which it got $msg$,  and stop sending messages.

A cycle is detected if $u$ gets back a ``cycle'' message. If there
is no cycle then $u$ gets a ``no-cycle'' message back from all its
in-neighbors.

\smallskip

\item{\bf Forward search:}
 If $\ell(u)\llex \ell(v)$ then $v$ sends the message
 $\ell(u)$ to its out-neighbors.
When a vertex $w$ gets $\ell(u)$ over an edge $e$ then it
performs one of the followings. \\
1) If $w$ also got a message $msg$ during
 the backward search it sends a ``cycle'' message back over $e$. \\
2) If  $\ell(u)\llex \ell(w)$  and $w$ gets $\ell(u)$ for the first time
then it sends $\ell(u)$  to all its out-neighbors.\\
3) If  $\ell(u)\geqlex \ell(w)$
 or $w$ has no
out-neighbors then $w$ sends a ``no-cycle'' message back on $e$.

When $w$ gets a ``cycle'' message  for the first time then it
forwards it to one of its in-neighbor from which it got $\ell(u)$
and stop sending messages. When $w$ gets a ``no-cycle'' message from
all of its out-neighbors then it sends a ``no-cycle'' message to all
its in-neighbors from which it got $\ell(u)$ and stop sending
messages.

A cycle is detected if $v$ gets back a ``cycle'' message. If there
is no cycle then $v$ gets a ``no-cycle'' message back from all its
out-neighbors.

\end{itemize}

\emph{Update($x,y$):}
\begin{itemize}
\item Let $\zeta$ be such that  $\ell(x) = \mathrm{LCP}(x,y) \| \zeta$, {\sl i.e.}, the label of $x$ is the concatenation of the longest common prefix with $y$, followed by $\zeta$.
\item Let $\zeta'$ be the longest prefix of $\zeta$ such that $r(\zeta'_j) < r(y)$ for $1\leq j \leq |\zeta'|$.
\item \label{alg:update} If $y$ is ranked set $\ell(y) = \mathrm{LCP}(x,y) \| \zeta' \| y$, otherwise set $\ell(y) = \mathrm{LCP}(x,y) \| \zeta'$.
\item\label{update:recurse} If $\ell(y)$ was updated then for all arcs $(y,w)\in E$, recursively apply update($y,w$).
\end{itemize}

We start by proving the correctness of the  Algorithm Insert$(u,v)$
assuming that it halts. We show that it halts and  bound the number
of messages that it sends in Section \ref{sec:analysis}.

\begin{lem}
Given a correct labeling of the vertices of an acyclic graph,
\begin{itemize}
  \item Insert($u,v$) will detect a cycle if the insertion of $(u,v)$ creates one.
  \item Insert($u,v$) will produce a correct labeling if the insertion of arc $(u,v)$ into the graph does not create a cycle.
\end{itemize}
\end{lem}
\begin{proof}

We start by showing that the algorithm detects a cycle if and only
if the insertion of $(u,v)$ creates one. If $\ell(u) \glex \ell(v)$
then from Theorem~\ref{T-no-path} it follows that there is no path
from $v$ to $u$. Thus, the insertion of $(u,v)$ does not create a
cycle, and the algorithm is correct. Consider now the case that
$\ell(u)=\ell(v)$. In this case it follows from
Theorem~\ref{T-no-path} that if there is a path from $v$ to $u$ then
all its vertices must have the same label and thus if a cycle exists
all its vertices have the same label. The algorithm detects such a
cycle during the backward search which would reach $v$. In this case
$v$ will send a cycle message that will reach $u$. If such a cycle
does not exist the search terminates at vertices with no
in-neighbors or with a smaller label. In both cases a ``no-cycle''
message will be sent back by each such vertex. Eventually $u$ will
get a ``no-cycle'' message from all its in-neighbors.

Consider now the case that $\ell(u)\llex \ell(v)$ and there is a
path (or more) from $v$ to $u$. Let $w$ be a vertex on such a path
$p$ for which $\ell(w)=\ell(u)$ for the first time when traversing
$p$ from $v$. Since the path terminates at $u$,  $w$ must exist. By
the definition of the backward search  $w$ gets the message
$msg=\langle v,\ell(u)\rangle$ during the backward search from $u$.  By the
definition of the forward search  $w$ will also get a message
 $\ell(u)$  during the forward search and will send back a ``cycle'' message that will reach $v$.
  If there is
no cycle then the forward search terminates at vertices $w$ such
that either $\ell(w)\geqlex \ell(u)$ or $w$ has no out-neighbors. In
both cases a ``no-cycle'' message is sent back by each such vertex.
Eventually, $v$ gets  a ``no-cycle'' message from each of its
out-neighbors.

We now turn to show that the labels are correctly maintained. For
every $w\in V$ let $\ell_{old}(w)$ be the label of $w$  before the
insertion of $(u,v)$ and let $\ell_{new}(w)$ be the correct label of
$w$ after the change. We show that $\ell_{new}(w)$ is indeed the
label of any $w\in V$ when the algorithm halts.

First notice that the predecessor set, $P(w)$, of each vertex $w$
uniquely defines its label. So for vertices $w$ such that $P(w)$
does not change by adding $(u,v)$ we should have
$\ell_{new}(w)=\ell_{old}(w)$.

By its definition, the algorithm changes the label only of vertices
$w$ such that $v\in P(w)$. So if  $v\not\in P(w)$ then
$\ell_{new}(w)=\ell_{old}(w)$ is the label of $w$ after the
insertion as required.

Let $v=w_1,w_2,\ldots, w_t$ be the vertices in $S(v)$ ordered by a
topological order. We prove by induction on this order that the
labels are correct. The basis of the induction holds for $v=w_1$ as
the algorithm updates $v$ if $\ell_{old}(u)\llex \ell_{old}(v)$ and
the update is correct by the definition of the labels. Assume that
when the algorithm updated the label  of $w\in \{ w_1,w_2,\ldots,
w_{i-1}\}$ for the last time then the label of $w_j$ was
$\ell_{new}(w_j)$, for each $1\le j\le i-1$. We prove that the
algorithm updates the label of $w_i$ to $\ell_{new}(w_i)$. Let
$\{w_{j_1},w_{j_2},\ldots, w_{j_r}\}\subseteq \{ w_1,w_2,\ldots,
w_{i-1}\}$ be all the in-neighbors of $w_i$ that are in $S(v)$. If
$\ell_{new}(w_i)\not= \ell_{old}(w_i)$  then  at least one of
$\{w_{j_1},w_{j_2},\ldots, w_{j_r}\}$ changed its label as well. By
the induction and the definition of the algorithm, each of
$\{w_{j_1},w_{j_2},\ldots, w_{j_r}\}$  eventually transmits its
correct new label to $w_i$ and as the update  procedure implements
the label definition the claim follows.

\end{proof}

%

\subsection{Analysis of the number of messages required}
\label{sec:analysis}

For a vertex $u$, an update to $\ell(u)$ means that the value of
$\ell(u)$ has changed. This is distinct from the number of update
messages to $u$, because update messages may have no effect on
$\ell(u)$. Note that the insertions of a single arc may produce
several updates to $\ell(u)$, this is because updates propagate
through the network at different rates.

Following an insertion of an arc $e$ (that did not close a cycle)
many calls to the update procedure are made. We define the {\em
schedule} of an arc $e$ to be the chronologically ordered sequence
of these calls (breaking ties arbitrarily). The schedule (and it's
length) depends on the arbitrariness of the choices in line
\ref{update:recurse} of Update($x,y$) above, and by variable message
timing in distributed environments. The length of the schedule is
the total number of messages sent in order to update the labels
following an arc insertion. Note that not every such message can
cause a label update at the target node so the total number of label
changes may be smaller than the schedule.


Fix the set of ranked vertices $L$ and the assignment $r$ of ranks
to vertices in $L$. Consider any vertex $u$. Let $\sigma^p(u)$ be
the set of all sequences of arc insertions, such that: after
inserting the arcs in $\sigma \in \sigma^p(u)$,  $u$ has  $p$ ranked
predecessors. Define $T_p(r,u)$ to be the maximal number of updates
to $\ell(u)$, for any sequence of arc insertions $\sigma \in
\sigma^p(u)$, and any schedule of updates for each of these
insertions.

 Let $\sigma \in \sigma^p(u)$ be the sequence of arcs $e_1, e_2, \ldots, e_{|\sigma|}$. Let $S(e_i)$, $1 \leq i \leq |\sigma|$ be the (possibly empty) set of vertices that became predecessors to $u$ following
the insertion of arc $e_i$, but were not predecessors to $u$ after the insertion of $e_{i-1}$.
Let
$u_1, u_2, \ldots, u_{p}$,  be the ranked predecessors of $u$ after the insertion of all the arcs in $\sigma$, in the order in which they became predecessors of $u$. {\sl I.e.}, where the vertices of $S(e_i)$ are a consecutive subsequence of $u_1, \ldots, u_p$ and appear before the vertices of $S(e_{i'})$ for $i' > i$. The vertices of $S(e_i)$ are ordered arbitrarily.

Let $\pi$ be some topological ordering of the ranked predecessors of $u$ that is consistent with the final set of arcs.
Let $j_1, j_2, \ldots, j_p$ be a
permutation of $1,\ldots,p$ such that $u_{j_i}$ appears before
$u_{j_{i+1}}$ in the topological ordering induced by $\pi$. Define $\beta(i)$ to be such that $j_{\beta(i)}=i$.

\begin{lem}
\label{lem:mapping} Fix some $u$,
  let $\sigma^p(u)$, $T_p(r,u)$ be as above. Choose a worst case $\sigma\in \sigma^p(u)$: that is $\sigma$, in conjunction with appropriate schedules, maximizes the number of updates to $\ell(u)$.
  This defines the sequence  $u_1,u_2,\ldots,u_p$ of ranked predecessors of $u$ as defined above.
  Let $u_\alpha$ be the ranked predecessor of $u$ of minimal rank. Then, $$T_p(r,u) \leq T_{\alpha-1}(r,u) + T_{p-\beta(\alpha)}(r,u) + 1.$$
\end{lem}
\begin{proof}
  Let $\gamma$ be such that $u_{\alpha} \in S(e_\gamma)$.
  We split the updates to $\ell(u)$ into three chronologically consecutive groups:
  \begin{itemize}
    \item Updates to $\ell(u)$ from the insertion of arcs $e_1, e_2, \ldots, e_{\gamma-1}$. Let $q=\sum_{k=1}^{\gamma-1} |S(e_k)|$, there are no more than $T_q(r,u) \leq T_{\alpha-1}(r,u)$ updates to $\ell(u)$ associated with these insertions.
    \item The first update to $\ell(u)$ subsequent to the insertion of arc $e_\gamma$.
    \item \label{split:hard} All subsequent updates to $\ell(u)$, let the number of such updates be denoted by $Z$. To prove this lemma we need to show that
    $Z\leq T_{p-\beta(\alpha)}(r,u)$.
  \end{itemize}

  We now consider a new graph consisting of $|V|+|L|$ vertices, and initially containing no edges.
  For every ranked vertex $v\in L$ we add a new (unranked) vertex $v'$ to $V$, let $V'= V \cup \{v' \mid v\in L\}$, the rank of $v$ remains unchanged.
   We build a sequence of arc insertions, $\tau$, (arcs between vertices of $V'$), and appropriate schedules, such that $u$ has no more than
   $p-\beta(\alpha)$ ranked predecessors, and the number of updates to $u$ in $\tau$ is $\geq Z$.

   The sequence $\tau$ is as follows:
   \begin{enumerate}
   \item Set $S=\emptyset$.
   \item For every arc $e=(x,y)\in \sigma$ where $x$ is not ranked add arc $(x,y)$ to $S$.
   \item For every arc $e=(x,y)\in \{ e_1, e_2, \ldots, e_\gamma \}$ such that \begin{enumerate} \item  $x\neq u_\alpha$ and $x$ is ranked, and, \item  $x$ is reachable from $u_\alpha$ (after adding $e_1, \ldots, e_\gamma$), and, \item $u$ is reachable from $y$ (after adding $e_1, \ldots e_\gamma$):
   \end{enumerate}  add arc $(x',y)$ to $S$.
   \item Inserting the arcs of $S$, in any order, never updates $\ell(u)$, nor do they introduce ranked predecessors to $u$. Let $\tau$ be insertions of the arcs of $S$ in some arbitrary order.
   \item \label{tau:ui} Following the arcs above we add arcs $(u_i,u_i')$ to $\tau$, for all ranked predecessors of $u$, $u_i$, in order of decreasing rank (not ordered by $i$).
   \item Subsequently, we add arcs $e_j$, $j>\gamma$ to $\tau$, if $e_j=(x,y)$ and $x$ is reachable from $u_\alpha$. These arcs appear in the same order as in $\sigma$.
   \end{enumerate}

We claim that if for every arc in $\tau$ we use the worst case
schedule (resulting in the maximal number of updates to $\ell(u)$)
then the number of label updates is at least $Z$.

Consider the updates to $\ell(u)$, as a consequence of inserting the
arc $e_\gamma$ in the original graph. The updated values of
$\ell(u)$ are all of the following form $u_\alpha=u_{k_0}, u_{k_1},
u_{k_2}, \ldots, u_{k_t}$ where these vertices lie along a path from
$u_{\alpha}$ to $u$, and $u_{k_i}$ is the vertex of minimal rank
along the subpath from $u_{k_{i-1}}$ to $u$.
   We can classify such updates to $\ell(u)$ according to the rank of $u_{k_1}$. When we add the arc $(u_{k_1}, u_{k_1}')$ to the new graph, we generate updates to $\ell(u)$ with labels that start with $u_{k_1}$.

Specifically, consider all changes of $\ell(u)$ (in the original
graph following the insertion of $e_\gamma$) to a label with
$u_\alpha$ as the first vertex and $u_{k_1}=u_j$ as the second
vertex for some fixed $u_j$. Every such label corresponds a path $Q$
as above, all ranked vertices along $Q$ (excluding $u_\alpha$) have
rank greater than the rank of $u_j$. In the new graph, when adding
the edge $(u_j, u_j')$ there is a path analogous to the path $Q$ in
which each ranked vertex $z$ is replaced by the edge $(z,z')$. So we
construct the following schedule for $(u_j, u_j')$.
\begin{enumerate}
\item
Consider a message from $x$ to $y$ in the schedule of $e_\gamma$
with a label $\ell$ containing $u_j$ as the second vertex (following
$u_\alpha$). In the schedule of $(u_j, u_j')$  we send a message
with a label equal to $\ell$ with $u_\alpha$ removed from $x'$ to
$y$. If the label of $y$ changes as a result of receiving this
message from $x'$ then $y$ sends a message to $y'$ containing its
new label. We send these messages in the same relative order as of
their corresponding messages in the schedule of $e_\gamma$.
 Each message from
a vertex $v$ is sent following a change in the label of $v$ since
this was the case in the schedule of $e_\gamma$.
\item
We continue the schedule arbitrarily until all labels are
consistent.
\end{enumerate}
The first part of this schedule generates an update to $\ell(u)$ for
every update to $\ell(u)$ with a label whose second vertex is $u_j$
that was generated by the schedule of $e_\gamma$. Thus,  for all
arcs $(u_j,u_j')$ together we generate at least as many updates
caused by the insertion of $e_\gamma$.

For each insertion of an arc $e_{\gamma+1}, e_{\gamma+2},\dots$ the
worst case schedule which we use runs over the same subgraph as the
subgraph used by the schedule of the  original insertion which each
ranked vertex replaced by an arc. Therefore it generated at least as
many updates.
    \end{proof}

Our next goal is to show the following:

\begin{lem} \label{lem:changes}
For all $u\in V$, $$E_r(T_p(u,r)) \in O(p).$$
\end{lem}

\begin{proof}

Let $\alpha=\argmin_i r(u_i)$. As the ranks are assigned randomly,
we have that for all $i=1, \ldots, p$, $\mbox{Prob}(\alpha=i) = 1/
p$. By Lemma \ref{lem:mapping} we have that
\begin{eqnarray*} E_r(T_p(r,u)) &=& \frac{1}{p}\cdot \sum_{\alpha=1}^p \left( E_r(T_{\alpha-1}(r,u)) + E_r(T_{p-\beta(\alpha)}(r,u)) + 1\right) \\
&=& 1 + \frac{2}{p} \cdot \sum_{\alpha=1}^p E_r(T_{\alpha-1}(r,u)). \end{eqnarray*}

Let $T_p(u) = E_r(T_p(r,u))$, we prove by induction that $T_p(u) \leq 2p$, Assuming $T_{j}(u) \leq 2j$ for all $0 \leq j \leq p-1$, we get that
\vspace*{-6pt}
    \begin{eqnarray*}
    T_p(u) \leq 1 +  \frac{2}{p}\sum_{j=0}^{p-1} T_j(u)
      \leq 1 + \frac{2}{p} \sum_{j=1}^{p-1} 2j
      = 1 + \frac{4}{p} \cdot p(p-1)/2
      = 1 +  2(p-1) \leq 2p.
    \end{eqnarray*}
\vspace*{-10pt}
 \end{proof}

We are now ready to bound the total amount of messages.

\begin{thm} \label{thm:32}
For an appropriate choice of a $q$-labeling, the expected total
number of messages that the algorithm described above sends is
$O(m^{3/2}\sqrt{\log n})$. Each message contains $O(\log^2 n)$ bits
with high probability. It takes $O(\log n)$ time to process a
message with high probability.
\end{thm}
\begin{proof}
We prove the lemma for a constant degree graph (which in particular
implies that $m=\Theta(n)$) and then indicate the changes required
to extend the proof to general graphs.
 For constant degree graphs to minimize the number of messages we use a $q$-labeling with
$q=\sqrt{\log n/n}$ which implies that $\E(\lambda) = \E(|L|) = nq =
\sqrt{n \log n}$. By Lemma \ref{lem:changes}, for any vertex $u$ the
expected number of updates to $\ell(u)$ is at most $O(\sqrt{ n\log
n})$. Furthermore, by Lemma \ref{lem:backsize} the
 number of vertices reached during a backward search from
$u$ is $O(\log n/q)= O(\sqrt{ n \log n})$ with high probability.

Consider the insertion of all edges except the last if it closes a
 cycle. Since our graph is of constant degree
the  number of messages sent by the backward search initiated by
each such insertion is $O(\sqrt{ n \log n})$. Therefore the number
of messages sent by all backward searches is $O(n\sqrt{ n \log n})$.
The forward search initiated by such an insertion traverses the same
edges that the following label-update process traverses. So the
total number of messages of forward searches equals  the total
number of messages required to update the labels which is $O(n\sqrt{
n \log n})$. We conclude that the total number of messages for
backward searches, forward searches and label updates is $O(n\sqrt{
n \log n})$.


Consider now the last insertion if it closes a cycle. The backward
search of this insertion also traverses $O(\sqrt{ n \log n})$
 vertices and sends $O(\sqrt{ n \log n})$ messages. The forward search of
 this insertion traverses each edge at most twice so it sends $O(m)$
 messages. We do not update the labels in this case.

To handle the general case of arbitrary indegrees we slightly change
the labeling as follows. We define a $q$-arc labeling to be an
assignment of ranks to vertices obtained as follows: Each arc
chooses to be ranked with probability $q$, if the arc is ranked then
it chooses a random rank so that with high probability the ranked
vertices have unique ranks and the ranks are small. The rank of a
vertex is the minimum rank of its ranked incoming arcs, if any.

One can modify proofs that depend on the number of ranked
predecessors ({\sl e.g.}, $T_p$) to depend on the number of ranked
predecessor arcs. In Lemma \ref{L-set-size} the number of ranked
predecessors (vertices) goes down by a constant factor. In the
arc-ranked variant of this Lemma, the number of incoming ranked arcs
to predecessor vertices goes
 down by a constant factor. In Lemma \ref{lem:backsize}, rather
 than bound the number of vertices in $D(v_t,v)$ we can bound the
 number of incoming arcs to vertices of $D(v_t,v)$.

 Using these modified Lemma and a $q$-arc labeling with $q\approx
 1/\sqrt{m}$ we get the $m^3/2\sqrt{\log n}$ bound.

{\bf Remark:} If we want arcs to be ranked with probability
$1/\sqrt{m}$, then $m$ has to be known (or at least approximately
known).
 In the distributed setting, this can  be justified using
standard techniques to recompute and distribute $m$ whenever it
doubles. Let $m_{\mathrm{old}}$ be the current estimate of the
number of arcs. If following an insertion of an arc $(u,v)$, vertex
$u$ initiates a recount of the arcs with probability
$1/m_{\mathrm{old}}$ then indeed each vertex would know $m$
approximately up to a factor of $2$. (There are other deterministic
ways to achieve this.)
\end{proof}

When doing a backwards search from $u$, we need only $O(1)$ time per
vertex by  maintaining for each vertex $v$ a list of its immediate
predecessors that have the same label as $v$. This can be maintained
over time by having $v$, whenever it changes its label, store a list
of all immediate predecessors that sent it an update message with
this new label.

In the theorem above we choose $q$ so as to minimize the number of
messages. To minimize time, and assuming that backward searches take
$O(1)$ time per vertex, we choose slightly different $q$ and get the
time bounds stated in the introduction.

We also obtain the following theorem by observing the modified
backward search requires $O((\log n / q)^2)$ time, this follows from
 Corollary~\ref{cor:backsize}.

\begin{thm} \label{thm:33}
Using a $q$- (vertex) labeling with $q= \sqrt[3]{\log n /n}$, and
the modification to the backward search above, the expected running
time of the algorithm is $O(mn^{2/3}(\log n)^{4/3})$.
\end{thm}

\section{Using queues to improve forward propagation} \label{nsquare:sec}

Every vertex $w$ maintains  a value $\ell^w(u) \geqlex \ell(u)$ for
all $u$ such that $(w,u)\in E$. The value $\ell^w(u)$ is the label
of $u$ when $w$ last communicated with $u$. {\sl I.e.}, whenever
vertex $w$ sends an update message to $u$, it receives in return the
current label of $u$ and updates $\ell^w(u)$. Note that $u$ may
update $\ell(u)$ following the message from $w$, or not, but in any
case it sends back to $w$ the (possibly new) $\ell(u)$. Every vertex
$w$ maintains a priority queue ordered by $\ell^w(u)$ for all $u$
such that $(w,u)\in E$.

We modify the propagation algorithm as follows: \begin{itemize} \item When vertex $w$ updates its label from $\ell_{{\rm old}}(w)$ to $\ell_{{\rm new}}(w)$,  ($\ell_{{\rm new}}(w)\llex \ell_{{\rm old}}(w)$), $w$ sends a message containing $\ell_{{\rm new}}(w)$ to all vertices $u$ in the priority queue such that $\ell^w(u) \glex \ell_{{\rm new}}(w)$. Such messages from $w$ to $u$ are called {\sl update messages}.
\item When vertex $u$ receives an update message from $w$ with $\ell_{{\rm new}}(w)$, $u$ will
 update it's own label $\ell_{\rm old}(u)$ if $\ell_{{\rm new}}(w)\llex \ell_{old}(u)$ and transmits the (possibly new) $\ell_{\rm new}(u)$ to $w$.
 \item Vertex $w$ sets $\ell^w(u)=\ell_{\rm new}(u)$ and updates the priority queue accordingly.
 \end{itemize}

We apply the algorithm of Section \ref{sec:2directions} with this
modified propagation method and with all vertices ranked. Since when
all vertices are ranked each vertex has a different label the
backward search from $v$ degenerates and contains only $v$. The
total number of messages required by the update procedure to update
the labels is $O(mn)$. The following theorem gives an
 upper bound of $O(n^2 \log n)$ on the total number of messages of
 the modified algorithm which is better for dense graphs.

\begin{thm}
If all vertices are ranked then the modified algorithm described
above sends $O(n^2\log n)$ messages on average. Each message
consists
 of $O(\log^2 n)$ bits with high probability, and it takes $O(\log
 n)$ time to process a message with high probability.
\end{thm}
\begin{proof}
Although we run the algorithm with all vertices ranked, our analysis
is more general and bounds the number of update messages send to
each vertex $u$ as a function of the number of ranked predecessors
$u$ has at the end.

We use the notations and definitions of Section \ref{sec:analysis}
with the following modifications. We define $\sigma^p(u)$ to be the
set of insertion sequences such that $u$ ends up with $p$ ranked
predecessors and with at most $p$ incoming neighbors. We define
 $T_p(r,u)$ to be
the maximal number of update messages sent  to $u$ (we count both
those that trigger a change of $\ell(u)$ and futile ones that do
not), for any sequence of arc insertions $\sigma \in \sigma^p(u)$,
and any schedule of updates for each of these insertions.

We prove that  $T_p(r,u)$ satisfies the recurrence
$$T_p(r,u) \leq T_{\alpha-1}(r,u) + T_{p-\beta(\alpha)}(r,u) + p.$$
The solution to this recurrence is $O(p\log p)$ so by summing up
over all vertices and substituting the worst case $p=n$ the theorem
follows.

Fix some $u$, and consider the worst case $\sigma\in \sigma^p(u)$:
that is $\sigma$, in conjunction with appropriate schedules,
maximizes the number of update messages to $u$. Let
 $u_1,u_2,\ldots,u_p$ be the final ranked predecessors of $u$ in the order
 in which they became predecessors of $u$ (as in Section \ref{sec:analysis}).
  Let $u_\alpha$ be the ranked predecessor of $u$ of minimal rank.
Let $\gamma$ be such that $u_{\alpha} \in S(e_\gamma)$.
  We split the update messages sent to $u$ into three chronologically consecutive groups:
  \begin{itemize}
    \item Updates from the schedules of the arcs $e_1, e_2, \ldots, e_{\gamma-1}$. Let $q=\sum_{k=1}^{\gamma-1} |S(e_k)|$, there are no more than $T_q(r,u) \leq T_{\alpha-1}(r,u)$
    update messages sent to $u$ associated with these insertions.

    \item Each in-neighbor $v$ of $u$ may send  a single update
    message to $u$ such that before sending this message $\ell^v(u)$
    did not start with $u_\alpha$ and after sending this message $\ell^v(u)$
starts with
    $u_\alpha$. Since the number of in-neighbors
    is at most $p$ for every $\sigma \in \sigma^p(u)$ we get that there are
    at most $p$ messages sent to $u$ of this kind.

    \item \label{split:hard2} All other messages sent to $u$. These are the messages from the schedules of
    $e_\gamma, e_{\gamma+1},\dots$ that were not counted in the previous item. Let the number of such messages be denoted by $Z$. As in the proof of Lemma \ref{lem:mapping} we now show that
    $Z\leq T_{p-\beta(\alpha)}(r,u)$.
  \end{itemize}


As in the proof of Lemma \ref{lem:mapping} we consider a new graph
consisting of $|V|+|L|$ vertices (recall that $L$ is the set of
ranked vertices), and initially containing no arcs. For every ranked
vertex $v\in L$ we add a new unranked vertex $v'$, the rank of $v$
remains unchanged. We consider the sequence of arc insertions,
$\tau$,  defined in the proof of Lemma \ref{lem:mapping}. At the end
of this sequence $u$ has at most $p-\beta(\alpha)$ ranked
predecessors and at most $p-\beta(\alpha)$ incoming arcs. We claim
that the schedules for the insertions of $(u_j,u'_j)$ defined in the
proof of Lemma \ref{lem:mapping} are valid schedules. This
immediately implies that the insertions of these arcs generate as
many update messages to $u$ as were generated by the insertion of
$e_\gamma$.

Consider the schedule of $(u_j,u'_j)$. The first message in this
schedule is from $u_j$ to $u'_j$ and it updates the label of $u'_j$
to contain $u_j$ (i.e.\ to be the same as the label of $u_j$).
 To prove that the rest of this schedule  is
valid  we need to argue that if a message with a label $\ell$ that
starts with ``$u_\alpha,u_j$'' was sent from $x$ to $y$ by the
schedule of $e_\gamma$, then we can send message with the same label
with $u_\alpha$ removed from $x'$ to $y$. Each time a vertex $y$ is
updated it sends a message to $y'$ so that they always have the same
label (it is obvious that these messages can be sent since they
cause real updates).

So consider such a message $\ell$ that was sent from $x$ to $y$ by
the schedule of $e_\gamma$. Let $\ell'$ be the corresponding message
with $u_\alpha$ removed that is to be sent by the schedule of
$(u_j,u'_j)$.  In the schedule of $(u_j,u'_j)$ if
 $\ell^{x'}(y)$ does not start with $u_j$ then it
must start with a vertex with rank greater than the rank of $u_j$
(vertices of smaller rank still cannot reach any other vertex) which
implies that $\ell'$ is lexicographically smaller than
$\ell^{x'}(y)$ and therefore $\ell'$ can be sent. Otherwise, $x'$
has already sent to $y$ a message earlier in this schedule. In this
case $\ell'$ must be lexicographically smaller than $\ell^{x'}(y)$
because this schedule is a subsequence of the schedule of
$e_\gamma$.
\end{proof}

\bibliographystyle{plain}
\bibliography{bibcycle}

\begin{thebibliography}{10}

\bibitem{AjwaniF10}
D.~Ajwani and T.~Friedrich.
\newblock Average-case analysis of incremental topological ordering.
\newblock {\em Discrete Applied Mathematics}, 158(4):240--250, 2010.

\bibitem{AjwaniFM08}
D.~Ajwani, T.~Friedrich, and U.~Meyer.
\newblock An {O({\it n}$^{\mbox{2.75}}$)} algorithm for incremental topological
  ordering.
\newblock {\em ACM Trans. on Algorithms}, 4(4), 2008.

\bibitem{AlpernHRSZ90}
B.~Alpern, R.~Hoover, B.~K. Rosen, P.~F. Sweeney, and F.~K. Zadeck.
\newblock Incremental evaluation of computational circuits.
\newblock In {\em SODA}, pages 32--42, 1990.

\bibitem{BarnatBC05}
J.~Barnat, L.~Brim, and J.~Chaloupka.
\newblock From distributed memory cycle detection to parallel ltl model
  checking.
\newblock {\em Electr. Notes Theor. Comput. Sci.}, 133:21--39, 2005.

\bibitem{BFG:SODA09}
M.~A. Bender, J.~T. Fineman, and S.~Gilbert.
\newblock A new approach to incremental topological ordering.
\newblock In {\em SODA}, 2009.

\bibitem{BeFiGiTa11}
Michael~A. Bender, Jeremy~T. Fineman, Seth Gilbert, and Robert~Endre Tarjan.
\newblock A new approach to incremental cycle detection and related problems.
\newblock {\em CoRR}, abs/1112.0784, 2011.

\bibitem{Cohen:1997}
Edith Cohen.
\newblock Size-estimation framework with applications to transitive closure and
  reachability.
\newblock {\em J. Comput. Syst. Sci.}, 55(3):441--453, December 1997.

\bibitem{FleischerHP00}
L.~Fleischer, B.~Hendrickson, and A.~Pinar.
\newblock On identifying strongly connected components in parallel.
\newblock In {\em IPDPS}, volume 1800 of {\em Lecture Notes in Computer
  Science}, pages 505--511, 2000.

\bibitem{HKMST:Talg12}
B.~Haeupler, T.~Kavitha, R.~Mathew, S.~Sen, and R.~E. Tarjan.
\newblock Incremental cycle detection, topological ordering, and strong
  component maintenance.
\newblock {\em ACM Trans. on Algorithms}, 8(1), 2012.

\bibitem{harary1965structural}
F.~Harary, R.Z. Norman, and D.~Cartwright.
\newblock {\em Structural models: an introduction to the theory of directed
  graphs}.
\newblock Wiley, 1965.

\bibitem{KatrielB06}
I.~Katriel and H.~L. Bodlaender.
\newblock Online topological ordering.
\newblock {\em ACM Trans. on Algorithms}, 2(3):364--379, 2006.

\bibitem{DBLP:journals/corr/abs-0711-0251}
T.~Kavitha and R.~Mathew.
\newblock Faster algorithms for online topological ordering.
\newblock {\em CoRR}, abs/0711.0251, 2007.

\bibitem{KnuthS74}
D.~E. Knuth and J.~L. Szwarcfiter.
\newblock A structured program to generate all topological sorting
  arrangements.
\newblock {\em Inf. Process. Lett.}, 2(6):153--157, 1974.

\bibitem{LiuC07}
H.~Liu and K.~Chao.
\newblock A tight analysis of the katriel-bodlaender algorithm for online
  topological ordering.
\newblock {\em Theor. Comput. Sci.}, 389(1-2):182--189, 2007.

\bibitem{Marchetti-SpaccamelaNR96}
A.~Marchetti-Spaccamela, U.~Nanni, and H.~Rohnert.
\newblock Maintaining a topological order under edge insertions.
\newblock {\em Inf. Process. Lett.}, 59(1):53--58, 1996.

\bibitem{PearceK06}
D.~J. Pearce and P.~H.~J. Kelly.
\newblock A dynamic topological sort algorithm for directed acyclic graphs.
\newblock {\em ACM Journal of Experimental Algorithmics}, 11, 2006.

\bibitem{Tarjan72dfs}
R.~E. Tarjan.
\newblock Depth-first search and linear graph algorithms.
\newblock {\em SIAM J. Comput.}, 1(2):146--160, 1972.

\end{thebibliography}
\end{document}